\documentclass[conference]{IEEEtran}
\ifCLASSINFOpdf
   \usepackage[pdftex]{graphicx}
\else
\fi
\usepackage[cmex10]{amsmath}
\usepackage[tight,footnotesize]{subfigure}
\usepackage{amssymb}

\newtheorem{theorem}{Theorem}

\newtheorem{assumption}{Assumption}
\newtheorem{proposition}{Proposition}
\newtheorem{remark_temp}[theorem]{Remark}

\begin{document}
%
\title{Spacecraft Attitude Stabilization with Piecewise-constant Magnetic Dipole Moment}

\author{\IEEEauthorblockN{Fabio Celani}
\IEEEauthorblockA{Department of Astronautical, Electrical, and Energy Engineering\\
Sapienza University of Rome\\
00138 Rome, Italy\\
Email: fabio.celani@uniroma1.it}
}


%


\maketitle

\begin{abstract}
In actual implementations of magnetic control laws for spacecraft attitude stabilization, the time in which Earth magnetic field is measured must be separated from the time in which magnetic dipole moment is generated. The latter separation translates into the constraint of being able to genere only piecewise-constant magnetic dipole moment. In this work we present attitude stabilization laws using only magnetic actuators that take into account of the latter aspect. Both a state feedback and an output feedback are presented, and it is shown that the proposed design allows for a systematic selection of  the sampling period.
\end{abstract}



%
\IEEEpeerreviewmaketitle

\section{Introduction}
Spacecraft attitude control can be obtained by adopting several actuation mechanisms. Among them electromagnetic actuators are widely used for generation of attitude control torques on satellites flying low Earth orbits. They consist of planar current-driven coils rigidly placed on the spacecraft typically along three orthogonal axes, and they operate on the basis of the interaction between the magnetic moment generated by those coils and the Earth's magnetic field; in fact, the interaction with the Earth's field generates a torque that attempts to align the total magnetic moment in the direction of the field.
Magnetic actuators, also known as magnetorquers, have the important limitation that control torque is constrained to belong to the plane orthogonal to the Earth's magnetic field. As a result, sometimes a reaction wheel accompanies magnetorquers to provide full three-axis control (see \cite[Section 7.4]{sidi97}); moreover, magnetorquers are often used for angular momentum dumping when reaction or momentum-bias wheels are employed for accurate attitude control (see \cite[Section 7.5]{sidi97}). Lately,  attitude stabilization using only magnetorquers has been considered as a feasible option especially for low-cost micro-satellites and for satellites with a failure in the main attitude control system. In such scenario many control laws have been designed, and a survey of various approaches adopted can be found in \cite{Silani:2005kx}; in particular, Lyapunov based design has been adopted in
\cite{Arduini:1997kx, Wisniewski:1999uq, Lovera:2004vn, Lovera:2005fk, celani_aa15}. In those works 
feedback control laws that require measures of both attitude and attitude-rate (i.e. state feedback control laws) are proposed; moreover, in \cite{Lovera:2004vn}  and \cite{celani_aa15} feedback control algorithms  which need measures of attitude only (i.e. output feedback control algorithms) are presented, too.

All control laws presented in the cited works are continuos-time; thus, in principle they require a continuous-time measurement of the geomagnetic field and a continuos-time generation of coils's currents. However, in practical implementations the time in which Earth's magnetic field is measured (measurement time) has to be separated from the time in which coils' currents are generated  (actuation time). The latter separation is necessary because currents flowing in the spacecraft's magnetic coils generate a local magnetic field that impairs the measurement of the geomagnetic field; as a result, when the Earth magnetic field is being measured no currents must flow in the coils, and consequently no magnetic actuation is possible. 
As a result, in practical applications a periodic switch between measurement time and actuation time is implemented. Usually measurement time is set much shorter than actuation time, so that the spacecraft is left unactuated for a very short fraction of the period; thus, the measurement process can be modeled as instantaneous, and during each period the magnetic dipole moment is kept constant since only the measurement at the beginning of the period is available. Thus the separation  illustrated before translates into the constraint of being able to generate only constant magnetic dipole moments during each period. Then, magnetic control laws compatible with the latter constraint could be obtained by simply discretizing continuos-time control algorithms (see e.g. \cite{pulecchi_cst10}); however, here a different approach is proposed; it consists in taking into account of the latter constraint directly during the control design phase. It will be shown that the proposed approach leads to a systematic selection of the  sampling interval; the latter aspect represents an advantage with respect to discretization of continuous-time control laws where the sampling interval is usually determined by trial an error.  The proposed method is presented for both a state and an output feedback .

\subsection{Notations} For $x \in \mathbb{R}^{n}$, $\|x\|$ denotes the Eucledian norm of $x$; for matrix $A$, $\|A\|$ denotes the 2-norm of $A$. Symbol $I$ represents the identity matrix. For $a \in \mathbb{R}^3$, $a^{\times}$ represents the skew symmetric matrix 
\begin{equation} \label{skew_symmetric}
a^{\times} =
\left[
\begin{array}{rcl}
  0 & -a_3 & a_2\\
  a_3 & 0 & -a_1\\
  -a_2 & a_1 & 0
\end{array}
\right]
\end{equation} so that
for $b  \in \mathbb{R}^3$, multiplication $a^{\times}b$ is equal to the cross product $a \times b$.

\section{Modeling and control algorithms}\label{modeling}

In order to describe the attitude dynamics of an Earth-orbiting rigid spacecraft, and in order to represent the geomagnetic field, it is useful to introduce the following reference frames.

\begin{enumerate}
  \item \emph{Earth-centered inertial frame} $\mathcal{F}_i$. A commonly used inertial frame for Earth orbits is the Geocentric Equatorial Frame, whose origin is in the Earth's center, its $x_i$ axis is the vernal equinox direction, its $z_i$ axis coincides with the Earth's axis of rotation and points northward,  and its $y_i$ axis completes an orthogonal right-handed frame (see \cite[Section 2.6.1]{sidi97} ).
   \item \emph{Spacecraft body frame} $\mathcal{F}_b$. The origin of this right-handed orthogonal frame attached to the spacecraft, coincides with the satellite's center of mass; its axes are chosen so that the inertial pointing objective is having  $\mathcal{F}_b$ aligned with $\mathcal{F}_i$.
\end{enumerate}

Since the inertial pointing objective consists in aligning $\mathcal{F}_b$ to $\mathcal{F}_i$, the focus will be on the relative kinematics and dynamics of the satellite with respect to the inertial frame. Let 
$q=[q_1\ q_2\ q_3\ q_4]^{\textrm{T}}=[q_v^{\textrm{T}}\ q_4]^{\textrm{T}}$
with $\|q\|=1$ be the unit quaternion representing rotation of  $\mathcal{F}_b$ with respect to $\mathcal{F}_i$; then, the corresponding attitude matrix  is given by 
 \begin{equation} \label{A_bo}
C(q)=(q_4^2-q_v^{\textrm{T}} q_v)I+2 q_v q_v^{\textrm{T}}-2 q_4 q_v^{\times}
\end{equation}
(see \cite[Section 5.4]{wie98}).

Let
\begin{equation} \label{w}
W(q)=\dfrac{1}{2}\left[\begin{array}{c}q_4 I + q_v^{\times} \\-q_v^{\textrm{T}}\end{array}\right]
\end{equation}
Then the relative attitude kinematics is given by 
\begin{equation} \label{kinem_body}
\dot q = W(q) \omega
\end{equation} where $\omega \in \mathbb{R}^{3}$ is the angular rate of $\mathcal{F}_b$ with respect to $\mathcal{F}_i$ resolved in $\mathcal{F}_b$ (see \cite[Section 5.5.3]{wie98}).

The attitude dynamics in body frame can be expressed by 
\begin{equation} \label{euler_eq}
J \dot \omega =- \omega^{\times} J \omega + T
\end{equation} where $J \in \mathbb{R}^{3 \times 3}$ is the spacecraft inertia matrix, and $Q \in \mathbb{R}^{3}$ is the control torque expressed in $\mathcal{F}_b$ (see \cite[Section 6.4]{wie98}). 

The spacecraft is equipped with three magnetic coils aligned with the $\mathcal{F}_b$ axes which generate the magnetic  control torque
\begin{equation} \label{t_coils}
Q=m_{coils} \times B^b=-B^{b \times}\ m_{coils}
\end{equation} where $m_{coils} \in \mathbb{R}^3$ is the vector of magnetic dipole moments for the three coils, and $B^b$ is the geomagnetic field at spacecraft expressed in body frame $\mathcal{F}_b$. 

Let $B^i$ be the geomagnetic field at spacecraft expressed in inertial frame $\mathcal{F}_i$. Note that $B^i$ varies with time at least because of the spacecraft's motion along the orbit. Then
\begin{equation} \label{B_b}
B^b(q,t)=C(q)B^i(t)
\end{equation}  
which shows explicitly the dependence of $B^b$ on both $q$ and $t$.

Grouping together equations (\ref{kinem_body}) (\ref{euler_eq}) (\ref{t_coils}) the following nonlinear time-varying system is obtained
\begin{equation} \label{nonlin_tv}
\begin{array}{rcl}
  \dot q &=& W(q) \omega\\
  J \dot \omega &=& - \omega^{\times} J \omega -B^{b}(q,t)^{\times}\ m_{coils} \end{array}
\end{equation} in which $m_{coils}$ is the control input.

In order to design control algorithms, it is important to characterize the time-dependence of $B^b(q,t)$ which is the same as  characterizing the time-dependence of $B^i(t)$. Assume that the orbit is circular of radius $R$; then, adopting the so called dipole model of the geomagnetic field (see \cite[Appendix H]{wertz78}) we obtain  \begin{equation} \label{geomegnetic_model}
B^i(t)=\frac{\mu_m}{R^3}[3 ((\hat m^i)^{\textrm{T}} \hat{r}^i(t)) \hat{r}^i(t) - \hat m^i  ]
\end{equation}
In equation (\ref{geomegnetic_model}), $\mu_m=7.746\ 10^{15}$ Wb m
is the total dipole strength (see \cite{Rodriguez-Vazquez:2012fk}), $r^i(t)$ is the spacecraft's position vector resolved in $\mathcal{F}_i$, and $\hat r^i(t)$ is the vector of the direction cosines of $r^i(t)$; finally $\hat m^i$ is the vector of direction cosines of the Earth's magnetic dipole moment expressed in $\mathcal{F}_i$. Here  we set $\hat m^i=[0\ 0\ -1]^{\textrm{T}}$ which corresponds to setting the dipole's colelevation equal to $180^{\circ}$. Even if a more  accurate value for that angle would be $170^{\circ}$ (see \cite{Rodriguez-Vazquez:2012fk}), here we approximate it to $180^{\circ}$  since this will substantially simplify future symbolic computations.

Equation (\ref{geomegnetic_model}) shows that in order to characterize the time dependence of $B^i(t)$ it is necessary to determine an expression for  $r^i(t)$ which is the spacecraft's position vector resolved in $\mathcal{F}_i$.  Define a coordinate system $x_p$, $y_p$ in the orbital's plane whose origin is at Earth's center; then, the position of satellite's center of mass is clearly given by
\begin{equation} \label{x_p_y_p}
\begin{array}{rcl}
  x^p(t) &=& R \cos(nt + \phi_0)\\
  y^p(t) &=& R \sin(nt + \phi_0)
\end{array}
\end{equation} where $n$ is the orbital rate, and $\phi_0$ an initial phase. Then, coordinates of the satellite in inertial frame $\mathcal{F}_i$ can be easily obtained from (\ref{x_p_y_p}) using an appropriate rotation matrix which depends on the orbit's inclination $incl$ and on the right ascension of the ascending node $\Omega$ (see \cite[Section 2.6.2]{sidi97}). Plugging  into (\ref{geomegnetic_model}) the expression of the latter coordinates, an explicit expression for $B^i(t)$ can be obtained; it can be easily checked that $B^i(t)$ is periodic with period $2 \pi/n$. Consequently system (\ref{nonlin_tv}) is a periodic nonlinear system with period $2 \pi/n$.

As stated before, the control objective is driving the spacecraft so that $\mathcal{F}_b$ is aligned with $\mathcal{F}_i$. From (\ref{A_bo}) it follows that $C(q)=I$ for $q=[q_v^{\textrm{T}}\ q_4]^{\textrm{T}}=\pm \bar q$ where $\bar q= [0\ 0\ 0\ 1]^{\textrm{T}}$. Thus,  the objective is designing control strategies for $m_{coils}$ so that $q_v \rightarrow 0$ and $\omega \rightarrow 0$. 

In \cite{celani_aa15} both a static state feedback and a dynamic output feedback are presented; both feedbacks were obtained as modifications of those in \cite{Lovera:2004vn, Lovera:2005fk}. 

The static state feedback control  proposed in \cite{celani_aa15} is given by
\begin{equation} \label{m_coils_cc}
m_{coils}(t)= (B^{b}(q(t),t)^{\times})^{\textrm{T}} (\epsilon^2 k_1 q_v(t) + \epsilon k_2 \omega(t))
\end{equation} and it is shown that picking $k_1>0$, $k_2>0$, and choosing $\epsilon>0$ and small enough, local exponential stability of $(q,\omega)=(\bar q, 0)$ is achieved. 

The dynamic output feedback proposed in \cite{celani_aa15} is given by
\begin{equation} \label{output_feedback_law}
\begin{array}{rcl}
\dot \delta(t) &=& \alpha (q(t) - \epsilon \lambda \delta(t))\\
m_{coils}(t) &=& (B^{b}(q(t),t)^{\times})^{\textrm{T}} \epsilon^2 \left( k_1 q_v(t)  \right. \\
&& \left. + k_2 \alpha \lambda W(q(t))^{\textrm{T}} (q(t)-\epsilon \lambda \delta(t)) \right)  
\end{array}
\end{equation}
with $\delta \in \mathbb{R}^4$. It is shown that selecting $k_1>0$, $k_2>0$, $\alpha>0$, $\lambda>0$,  and choosing $\epsilon>0$ small enough, local exponentially stability of equilibrium $(q,\omega,\delta)=(\bar q, 0,\frac{1}{\epsilon \lambda} \bar q)$ is achieved.

Both control laws (\ref{m_coils_cc}) and (\ref{output_feedback_law}) seem implementable in practice since $B^b(q(t),t)$ can be measured using magnetometers, $q_v(t)$ can be measured using attitude sensors,  and $\omega(t)$ can be measured by attitude rate sensors. However, as explained in the introduction, the time in which $B^b$ is measured should be separated from the time in which control action is applied; in fact, magnetic control torque is obtained by generating currents flowing through magnetic coils, and those currents create a local magnetic field which impairs the measurement of $B^b$. As a result, in order to take into account of the previous constraint, we should consider only feedback laws in which $B^b$ is instantaneously measured at the beginning of some fixed length interval, $m_{coils}$ is held constant over that interval, and the latter operations are repeated periodically. In the sequel $T$ will denote the sampling interval, that is the length of  each of those intervals.

 A static state feedback law that fulfills the previous constraint can be designed by simply discretizing (\ref{m_coils_cc}) through a sample-and-hold operation, thus obtaining the following control algorithm 
\begin{multline} \label{m_coils_dc}
m_{coils}(t)= (B^{b}(q(kT),kT)^{\times})^{\textrm{T}} (\epsilon^2 k_1 q_v(kT) + \epsilon k_2 \omega(k T))\\ k=0,1,2,\ldots \ \ \ kT \leq t < (k+1) T
\end{multline}

Similarly, a dynamic output feedback law compatible with the constraint of generating a piecewise-constant signal for $m_{coils}$ can be obtained by discretizing (\ref{output_feedback_law}) using the forward differencing approximation (see \cite{fadali_visoli_12}), thus obtaining the following control law
\begin{equation} \label{output_feedback_law_discrete}
\begin{array}{rcl}
\delta((k+1)T) &=& \delta(kT)+T \alpha (q(kT) - \epsilon \lambda \delta(kT))\\[1mm]
m_{coils}(t) &=& (B^{b}(q(kT),kT)^{\times})^{\textrm{T}} \epsilon^2 \left( k_1 q_v(kT)  \right. \\[1mm]
&& \left. + k_2 \alpha \lambda W(q(kT))^{\textrm{T}} (q(kT)-\epsilon \lambda \delta(kT)) \right)\\[1mm]
&& k=0,1,2,\ldots \ \ \ kT \leq t < (k+1) T  
\end{array}
\end{equation}
with $\delta \in \mathbb{R}^4$.

\section{State-feedback Design}

In this section we will focus on state-feedback law (\ref{m_coils_dc}) and show how to choose the parameters $k_1$, $k_2$, $\epsilon$, and the sampling interval $T$  so that equilibrium $(q,\omega)=(\bar q, 0)$ is locally exponentially stable for closed-loop system (\ref{nonlin_tv}) (\ref{m_coils_dc}). In order to derive those design criteria, it suffices considering the restriction of closed-loop system (\ref{nonlin_tv}) (\ref{m_coils_dc}) to the open set $\mathbb{S}^{3+} \times \mathbb{R}^3$ where 
\begin{equation} \label{S3}
\mathbb{S}^{3+}=\{ q \in \mathbb{R}^4\  |\  \|q\|=1,\ q_{4}>0\} 
\end{equation} Since on the latter set $q_{4}=(1-q_{v}^{\textrm{T}} q_{v})^{1/2}$  then the considered restriction is given by the  following nonlinear time-varying hybrid system
\begin{multline} \label{closed_loop_reduced}
\begin{array}{rcl}
  \dot q_{v}(t) &=& W_v(q_{v}(t)) \omega(t) \\
  J \dot \omega(t) &=& - \omega(t)^{\times} J \omega(t) - [R_v(q_v(t)) B^i(t)]^{\times} m_{coils}(kT)
  \end{array}\\
 \ \ \ \ \ \ \ \ \ \ \ \ \ \ \ \ \ \ \ \ \ \ \ \ \ \ \ \ \ \ \ \ \ \ \ \ \ \ \ \ \ kT \leq t \leq (k+1)T\\[2mm]
 m_{coils}(kT)=[R_v(q_v(kT)) B^i(kT)]^{\times \textrm{T}} (\epsilon^2 k_1 q_v(kT) \\ + \epsilon k_2 \omega(kT))\\
 k=0,1,2,\ldots
\end{multline}
where
\begin{equation} \label{w_v}
W_v(q_{v})=\dfrac{1}{2} \left[ \left( 1-q_{v}^{\textrm{T}} q_{v} \right)^{1/2} I + q_{v}^{\times} \right]
\end{equation} and
 \begin{equation} \label{a_v}
R_v(q_{v})=\left( 1-2 q_{v}^{\textrm{T}} q_{v} \right) I+2 q_{v} q_{v}^{\textrm{T}}-2 \left( 1- q_{v}^{\textrm{T}} q_{v} \right)^{1/2} q_{v}^{\times}
\end{equation} and where (\ref{B_b}) has been used. It is simple to show that if $(q_v,\omega)=(0, 0)$ is a locally exponentially stable equilibrium for  (\ref{closed_loop_reduced}), then $(q,\omega)=(\bar q, 0)$ is a locally exponentially stable equilibrium for (\ref{nonlin_tv}) (\ref{m_coils_dc}).

Next, consider the linear approximation of (\ref{closed_loop_reduced}) about the equilibrium $(q_v,\omega)=(0, 0)$ as defined in  \cite{Mancilla-Aguilar:2003fk}. The latter approximation is given by
\begin{equation} \label{closed_loop_linearized}
\begin{array}{rcl}
  \dot q_{v}(t) &=& \frac{1}{2} \omega(t) \\
  J \dot \omega(t) &=&  - B^i(t)^{\times} m_{coils}(kT)
  \end{array}\\
 \ \ \ \ \ kT \leq t \leq (k+1)T
 \end{equation}
 \begin{multline} \label{m_coils_dt}
 m_{coils}(kT)=B^i(kT)^{\times T} (\epsilon^2 k_1 q_v(kT) + \epsilon k_2 \omega(kT))\\
 k=0,1,2,\ldots
\end{multline}
Note that since $B^i$ is bounded (see (\ref{geomegnetic_model})), then Theorem II.1 in \cite{Mancilla-Aguilar:2003fk} applies to the nonlinear time-varying hybrid system (\ref{closed_loop_reduced}); consequently, $(q_v,\omega)=(0, 0)$  is a locally exponentially stable equilibrium for (\ref{closed_loop_reduced}) if and only if it is an exponentially stable equilibrium for (\ref{closed_loop_linearized}) (\ref{m_coils_dt}).

Next consider  the continuous-time system (\ref{closed_loop_linearized}) and sample its state $[q_v(t)^{\textrm{T}} \omega(t)^{\textrm{T}}]^{\textrm{T}}$ at $t=(k+1)T$. The following discrete-time system is thus obtained
\begin{multline} \label{sampled_system}
\begin{array}{rcl}
  q_v((k+1)T) &=& q_v(kT)+\frac{T}{2} \omega(kT)\\[2 mm]
 && -J^{-1}G_1(k,T) m_{coils}(kT)\\[2 mm]
 \omega((k+1)T) &=& \omega(kT)-J^{-1}G_2(k,T) m_{coils}(kT)\\
\end{array}\\ k=0,1,2\ldots
\end{multline} where
\begin{equation*}
G_1(k,T)=\int_{kT}^{(k+1)T} \frac{1}{2}((k+1)T-\tau) B^i(\tau)^{\times} d \tau
\end{equation*}
and
\begin{equation*}
G_2(k,T)=\int_{kT}^{(k+1)T} B^i(\tau)^{\times} d \tau
\end{equation*} By  \cite[Proposition 7]{Iglesias:1995fj} it follows that if the linear time-varying discrete-time system (\ref{m_coils_dt}) (\ref{sampled_system}) is exponentially stable then the linear time-varying hybrid system (\ref{closed_loop_linearized}) (\ref{m_coils_dt}) is exponentially stable.

Based on the previous discussion, our objective has become studying stability of linear time-varying discrete-time system (\ref{m_coils_dt}) (\ref{sampled_system}). For that purpose it is useful to perform the following change of variables $z_1(k)=q_v(kT)$, $z_2(k)= \omega(kT)/\epsilon$  obtaining
\begin{equation} \label{z_system}
\begin{array}{rcl}
  z_1(k+1) &=& z_1(k)+\epsilon \frac{T}{2} z_2(k)- \epsilon^2 J^{-1} G_1(k,T)B^i(kT)^{\times T}\\[1 mm]&& ( k_1 z_1(k) +  k_2 z_2(k))\\[2 mm]
  z_2(k+1) &=& z_2(k)-\epsilon J^{-1} G_2(k,T)B^i(kT)^{\times T}\\[1 mm]
  &&( k_1 z_1(k) +  k_2 z_2(k))
\end{array}
\end{equation} 
Note that since $G_i(k,0)=0\ \ i=1,2$, then $G_i(k,T)$ can be expressed as 
\begin{equation}
G_i(k,T)=H_i(k,T) T\ \ i=1,2
\end{equation} where
\begin{equation*}
H_i(k,T)=\int_0^1 \frac{\partial G_i}{\partial T}(k,sT)ds\ \ i=1,2
\end{equation*} Let $L(k,T)=H_2(k,T) B^i(kT)^{\times T}$. Note that since $B^i(t)$ is periodic, then the following average
\begin{equation*}
L_{av}(T)=\lim_{N \rightarrow \infty} \frac{1}{N} \sum_{k=s+1}^{k=s+N} L(k,T)
\end{equation*} is well defined and indipendent of integer $s$. The expression of $L_{av}$ has been computed symbolically using Mathematica\textsuperscript{TM}, but here it is omitted to save space. It turns out that $L_{av}^0 \triangleq \lim_{T \rightarrow 0} L_{av}(T)$ is symmetric, and that $L_{av}^0$ is positive definite if and only if the orbit is not equatorial, that is if its inclination satisfies $incl \neq 0$. Thus, we make the following assumption.

\begin{assumption} \label{orbit_hypo}
The spacecraft's orbit satisfies condition $ L_{av}^0>0$.
\end{assumption}

At this point, consider the average system of (\ref{z_system}) as defined in \cite{Bai:1988fk} which is given by
\begin{equation} \label{average_system}
\begin{array}{rcl}
  z_1(k+1) &=& z_1(k)+\epsilon \frac{T}{2} z_2(k)\\[2 mm]
  z_2(k+1) &=& z_2(k)-\epsilon T J^{-1} L_{av}(T)( k_1 z_1(k) +  k_2 z_2(k))
\end{array}
\end{equation}
It can be verified that all assumptions of \cite[Theorem 2.2.2]{Bai:1988fk}  are fulfilled by system (\ref{z_system}). Thus, the following proposition holds true.

\begin{proposition} \label{prop_averaged}
If the average system (\ref{average_system}) is exponentially stable for all $0<\epsilon \leq \epsilon_0$, then there exists $0<\epsilon_1 \leq \epsilon_0$, such that system (\ref{z_system}) is exponentially stable for all $0<\epsilon\leq \epsilon_1$.
\end{proposition}

For the average system (\ref{average_system}) the following stability result holds.

\begin{theorem}  \label{stability_averaged}
Assume that $k_1>0$, $k_2>0$ and Assumption \ref{orbit_hypo} is satisfied. Then there exists $T^*>0$, and for every $0<T<T^*$, there exists $\epsilon_0>0$, such that fixed $0<T<T^*$ system  (\ref{average_system}) is exponentially stable for all $0<\epsilon \leq \epsilon_0$.
\end{theorem}

\begin{proof}    
Rewrite (\ref{average_system}) in the following compact form
\begin{equation} \label{averaged_compact}
z(k+1)=z(k)+\epsilon T A_s(T) z(k)
\end{equation} with
\begin{equation} \label{A_s}
A_s(T)=
\left[
\begin{array}{cc}
0 & \frac{1}{2}I\\[3mm]
-k_1 J^{-1} L_{av}(T) & -k_2 J^{-1} L_{av}(T)
\end{array}
\right]
\end{equation} It will be shown now that $A_s^0 \triangleq \lim_{T \rightarrow 0} A_s(T)$ is a Hurwitz matrix. For that purpose it is useful to introduce the continuous-time system $\dot w=A_s^0 w$ 
 which in expanded form reads as follows
\begin{equation} \label{linearized_time_varying}
\begin{array}{rcl}
  \dot w_1 &=& \dfrac{1}{2} w_2\\
   J \dot w_2 &=&  - L_{av}^0 ( k_1 w_1 +  k_2 w_2)
 \end{array}
\end{equation}
Introduce the Lyapunov function for (\ref{linearized_time_varying})
\begin{equation*}
V_1(w_{1},w_2)= k_1w_{1}^{\textrm{T}} L_{av}^0 w_{1}+\frac{1}{2}   w_2^{\textrm{T}} J  w_2
\end{equation*}
then
\begin{equation*}
\dot V_1(w_{1},w_2)=- k_2 w_2^{\textrm{T}} L_{av}^0 w_2
\end{equation*}
By La Salle's invariance principle \cite[Theorem 4.4]{khalil00}, it follows that  (\ref{linearized_time_varying}) is exponentially stable, and thus $A_s^0$ is Hurwitz. Then, by continuity there exists $T^*>0$ such that $A_s(T)$ is Hurwitz for all $0 < T<T^*$. Next, fix $0<T<T^*$, and consider symmetric matrix $P_s>0$ such that
\begin{equation*}
P_s A_s(T)+A_s(T)^{\textrm{T}} P_s = -I
\end{equation*} Let $V_2(z)=z^{\textrm{T}} P_s z$ be a candidate Lyapunov function for system (\ref{averaged_compact}). Then, the following holds
\begin{multline*}
\Delta V_2(z) \triangleq (z+\epsilon T A_s(T) z)^{\textrm{T}} P_s (z+\epsilon T A_s(T) z)-z^{\textrm{T}} P_s z\\
\leq \epsilon T (-1+\epsilon T  \| A_s(T)^{\textrm{T}} P_s A_s(T) \| ) \| z \|^2
\end{multline*} As a result, setting
\begin{equation} \label{epsilon_0}
\epsilon_0=\frac{1}{2 T \| A_s(T)^{\textrm{T}} P_s A_s(T) \|}
\end{equation}
 we obtain that system (\ref{averaged_compact}) is exponentially stable for all $0<\epsilon \leq \epsilon_0$.

\end{proof}

Based on the previous theorem, fix $0<T<T^*$, and consider the corresponding value of $\epsilon_0$ given by (\ref{epsilon_0}); from Proposition 2 it follows that there exists $0 < \epsilon_1 \leq \epsilon_0$ such that system (\ref{z_system}) is exponentially stable for all $0<\epsilon\leq \epsilon_1$. In conclusion, from Theorem \ref{stability_averaged}, Proposition \ref{prop_averaged}, and the preceding discussion, the following main proposition is obtained.

\begin{proposition} \label{main_prop}
\emph{(State feedback).} Consider the magnetically actuated spacecraft described by (\ref{nonlin_tv}) and apply the piecewise-constant static state-feedback (\ref{m_coils_dc}) with $k_1>0$ and $k_2>0$. Then, under Assumption \ref{orbit_hypo}, there exists $T^*>0$, and for every $0<T<T^*$, there exists $\epsilon_1 \leq \epsilon_0$, with $\epsilon_0$ given by (\ref{epsilon_0}), such that fixed $0<T<T^*$, for all $0<\epsilon \leq \epsilon_1$ equilibrium $(q,\omega)=(\bar q, 0)$ is locally exponentially stable for (\ref{nonlin_tv}) (\ref{m_coils_dc}).
\end{proposition}


\section{Output-feedback Design}

In this section we will focus on output-feedback law (\ref{output_feedback_law_discrete}) and show how to choose parameters $k_1$, $k_2$, $\alpha$, $\lambda$, $\epsilon$, and sampling interval $T$ so that equilibrium $(q,\omega,\delta)=(\bar q, 0,\frac{1}{\epsilon \lambda} \bar q)$  is locally exponentially stable for closed-loop system (\ref{nonlin_tv}) (\ref{output_feedback_law_discrete}).

Following the same approach as in state-feedback design, we determine the linear time-varying discrete-time system given by the interconnection of (\ref{sampled_system}) with the following system
\begin{equation} 
\begin{array}{rcl} \label{output_fdbk_discrete}
  \delta_v((k+1)T) &=& \delta_v(kT)+T \alpha (q_v(kT)-\epsilon \lambda \delta_v(kT))\\
  \tilde \delta_4((k+1)T) &=& \tilde \delta_4(kT)-T \alpha \epsilon \lambda \tilde \delta_4(kT)\\[2mm]
  m_{coils}(kT) &=& (B^{i}(q(kT),kT)^{\times}){\textrm{T}} \epsilon^2 \left[ k_1 q_v(kT)  \right. \\[1mm]
&& \left. + \frac{1}{2} k_2 \alpha \lambda  (q_v(kT)-\epsilon \lambda \delta_v(kT)) \right]
\end{array}
\end{equation} where $\delta_v=[\delta_1\ \delta_2\ \delta_3]^{\textrm{T}}$ and $\tilde \delta_4=\delta_4-\frac{1}{\epsilon \delta}$. Similarly to the state feedback case, if system (\ref{output_fdbk_discrete}) is exponentially stable, then $(q,\omega,\delta)=(\bar q, 0,\frac{1}{\epsilon \lambda} \bar q)$ is an exponentially stable equilibrium for (\ref{nonlin_tv}) (\ref{output_feedback_law_discrete}).

For system (\ref{output_fdbk_discrete}), perform the change of variables $z_1(k)=q_v(kT)$, $z_2(k)= \omega(kT)/\epsilon$, $z_3(k)=q_v(kT)-\epsilon \lambda \delta_v(kT)$, $z_4(k)= \tilde \delta_4(kT)$ obtaining
\begin{equation} \label{z_system_output}
\begin{array}{rcl}
  z_1(k+1) &=& z_1(k)+\epsilon \frac{T}{2} z_2(k) - \epsilon^2 J^{-1} G_1(k,T)\\[1mm]
  && (B^i(kT)^{\times})^{\textrm{T}} \left( k_1 z_1(k) +  \frac{1}{2}k_2 \alpha \lambda z_3(k) \right)\\[3 mm]
  z_2(k+1) &=& z_2(k)-\epsilon J^{-1} G_2(k,T)(B^i(kT)^{\times})^{\textrm{T}}\\[2mm]
  &&\left( k_1 z_1(k) +   \frac{1}{2}k_2 \alpha \lambda z_3(k) \right)\\[3mm]
  z_3(k+1) &=& z_3(k) + \epsilon T \left(  \frac{1}{2}z_2(k)-\alpha \lambda z_3(k)  \right)\\[3 mm]
 && - \epsilon^2 J^{-1} G_1(k,T)(B^i(kT)^{\times})^{\textrm{T}}\\[2mm]
 && \left( k_1 z_1(k) +  \frac{1}{2}k_2 \alpha \lambda z_3(k) \right)\\[3 mm]
  z_4(k+1) &=& z_4(k) - \epsilon T \alpha  \lambda z_4(k)
\end{array}
\end{equation} 
The average system of (\ref{z_system}) as defined in \cite{Bai:1988fk} is given by
\begin{equation} \label{average_system_output}
\begin{array}{rcl}
  z_1(k+1) &=& z_1(k)+\epsilon \frac{T}{2} z_2(k)\\[2 mm]
  z_2(k+1) &=& z_2(k)-\epsilon T J^{-1} L_{av}(T)( k_1 z_1(k) +  k_2 z_2(k))\\[2 mm]
  z_3(k+1) &=& z_3(k)+\epsilon T \left( k_1 z_1(k) +   \frac{1}{2}k_2 \alpha \lambda z_3(k) \right)\\[4 mm]
  z_4(k+1) &=& z_4(k) -  \epsilon T \alpha \lambda z_4(k)
\end{array}
\end{equation}
Then, a proposition parallel to Proposition \ref{prop_averaged} holds true, and  for the average system (\ref{average_system_output}) the following stability result holds.

\begin{theorem}  \label{stability_averaged_output}
Assume that $k_1>0$, $k_2>0$, $\alpha>0$, $\lambda>0$ and Assumption \ref{orbit_hypo} is satisfied. Then there exists $T^*>0$, and for every $0<T<T^*$, there exists $\epsilon_0>0$, such that fixed $0<T<T^*$ system  (\ref{average_system_output}) is exponentially stable for all $0<\epsilon \leq \epsilon_0$.
\end{theorem}

\begin{proof}    
Rewrite (\ref{average_system_output}) in the following compact form
\begin{equation} \label{averaged_compact_output}
z(k+1)=z(k)+\epsilon T A_o(T) z(k)
\end{equation} with
\begin{multline*}
A_o(T)=\\
\left[
\begin{array}{cccc}
0 & \frac{1}{2}I & 0 & 0\\[3mm]
-k_1 J^{-1} L_{av}(T) & 0 & - \frac{1}{2} k_2 \alpha \lambda J^{-1} L_{av}(T) & 0\\
0 & \frac{1}{2}I & -\alpha \lambda I & 0\\[2mm]
0 & 0 & 0 & -\alpha \lambda
\end{array}
\right]
\end{multline*} It will be shown now that $A_o^0 \triangleq \lim_{T \rightarrow 0} A_o(T)$ is a Hurwitz matrix. For that purpose it is useful to introduce the continuous-time system $\dot w=A_o^0 w$ 
 which in expanded form reads as follows
\begin{equation} \label{linearized_time_varying_output}
\begin{array}{rcl}
  \dot w_1 &=& \dfrac{1}{2} w_2\\
   J \dot w_2 &=&  - L_{av}^0 \left( k_1 w_1 +  \frac{1}{2} k_2 \alpha \lambda w_3 \right)\\
   \dot w_3 &=& \frac{1}{2} w_2 - \alpha \lambda w_3\\
   \dot w_4 &=& - \alpha \lambda w_4
 \end{array}
\end{equation}
Introduce the Lyapunov function for (\ref{linearized_time_varying_output})
\begin{equation*}
V_3(w)= k_1w_{1}^{\textrm{T}} L_{av}^0 w_{1}+\frac{1}{2}   w_2^{\textrm{T}} J  w_2 +\frac{1}{2} k_2 \alpha \lambda  w_3^{\textrm{T}} L_{av}^0  w_3 + \frac{1}{2} w_4^2
\end{equation*}
then
\begin{equation*}
\dot V_3(w)=- k_2 \alpha^2 \lambda^2 w_3^{\textrm{T}} L_{av}^0 w_3 - \alpha \lambda w_4^2
\end{equation*}
By La Salle's invariance principle \cite[Theorem 4.4]{khalil00}, it follows that  (\ref{linearized_time_varying_output}) is exponentially stable, and thus $A_o^0$ is Hurwitz. Then, the proof can be completed as in the state feedback case. 
In particular, fixing $T$ small enough so that $A_o(T)$ is Hurwitz, and letting $P_o>0$ such that
\begin{equation*}
P_o A_o(T)+A_o(T)^{\textrm{T}} P_o = -I
\end{equation*} we have that system  (\ref{averaged_compact_output}) is exponentially stable for all $0<\epsilon \leq \epsilon_0$ with
\begin{equation} \label{epsilon_0_output}
\epsilon_0=\frac{1}{2 T \| A_o(T)^{\textrm{T}} P_o A_o(T) \|}
\end{equation}

\end{proof}

Thus, the following main proposition holds. 

\begin{proposition} \label{main_prop_output}
\emph{(Output feedback).} Consider the magnetically actuated spacecraft described by (\ref{nonlin_tv}) and apply the piecewise-constant dynamic output-feedback (\ref{output_feedback_law_discrete}) with $k_1>0$, $k_2>0$, $\alpha>0$, and $\lambda>0$. Then, under Assumption \ref{orbit_hypo}, there exists $T^*>0$, and for every $0<T<T^*$, there exists $\epsilon_1 \leq \epsilon_0$, with $\epsilon_0$ given by (\ref{epsilon_0_output}), such that fixed $0<T<T^*$ for all $0<\epsilon \leq \epsilon_1$ equilibrium $(q,\omega,\delta)=(\bar q, 0,\frac{1}{\epsilon \lambda} \bar q)$ is locally exponentially stable for (\ref{nonlin_tv}) (\ref{output_feedback_law_discrete}).
\end{proposition}

\section{Case study}

We consider the same case study presented in \cite{celani_aa15}  in which the spacecraft's inertia matrix is given by
$J=\text{diag}[27\ 17 \ 25] \ \text{kg m}^2$; the inclination of the orbit is given by $incl=87^{\circ}$, and the orbit's altitude is 450 km; the corresponding 
orbital period is about 5600 s. Without loss of generality the right ascension of the ascending node $\Omega$ is set equal to 0, whereas the initial phase $\phi_0$ (see (\ref{x_p_y_p})) has been randomly selected and set equal to $\phi_0=0.94$ rad.

Consider an initial state characterized by attitude equal to to the target attitude $q(0)=\bar q$, and by the following high initial angular rate
\begin{equation} \label{omega_0}
\omega(0)=[0.02\ \ 0.02\ \ -0.03]^{\textrm{T}}\ \text{rad/s}
\end{equation} due for example to an impact with an object.

In \cite{celani_aa15} the continuous-time feedback (\ref{m_coils_cc}) has been designed setting $k_1=2 \ 10^{11}$, $k_2=3\ 10^{11}$, and $\epsilon=10^{-3}$. Here, we keep $k_1=2 \ 10^{11}$, $k_2=3\ 10^{11}$ and parameters $T$ and $\epsilon$ are chosen as follows. First, studying numerically the behavior of the eigenvalues of  $A_s(T)$ with $T$ (see (\ref{A_s})), we determine the value $T^*=1490$ s for which it occurs that $A_s(T)$ is Hurwitz for all $0 < T < T^*$. Then, the sampling time has been set equal to $T=20$ to which there corresponds $\epsilon_0=1.3\ 10^{-3}$ (see (\ref{epsilon_0})). By Proposition \ref{main_prop}, it occurs that setting $\epsilon \leq 1.3\ 10^{-3}$ and small enough, equilibrium $(q,\omega)=(\bar q, 0)$ is locally exponentially stable for (\ref{nonlin_tv}) (\ref{m_coils_dc}). Here, proceeding by trial and error, we fix $\epsilon=10^{-3}$. Simulations results reported in Fig. \ref{evolutions_state_fdbck} show that actually acquisition of the desired attitude is achieved.
\begin{figure}[h] 
\centering
\subfigure{\includegraphics[width=8.4cm]{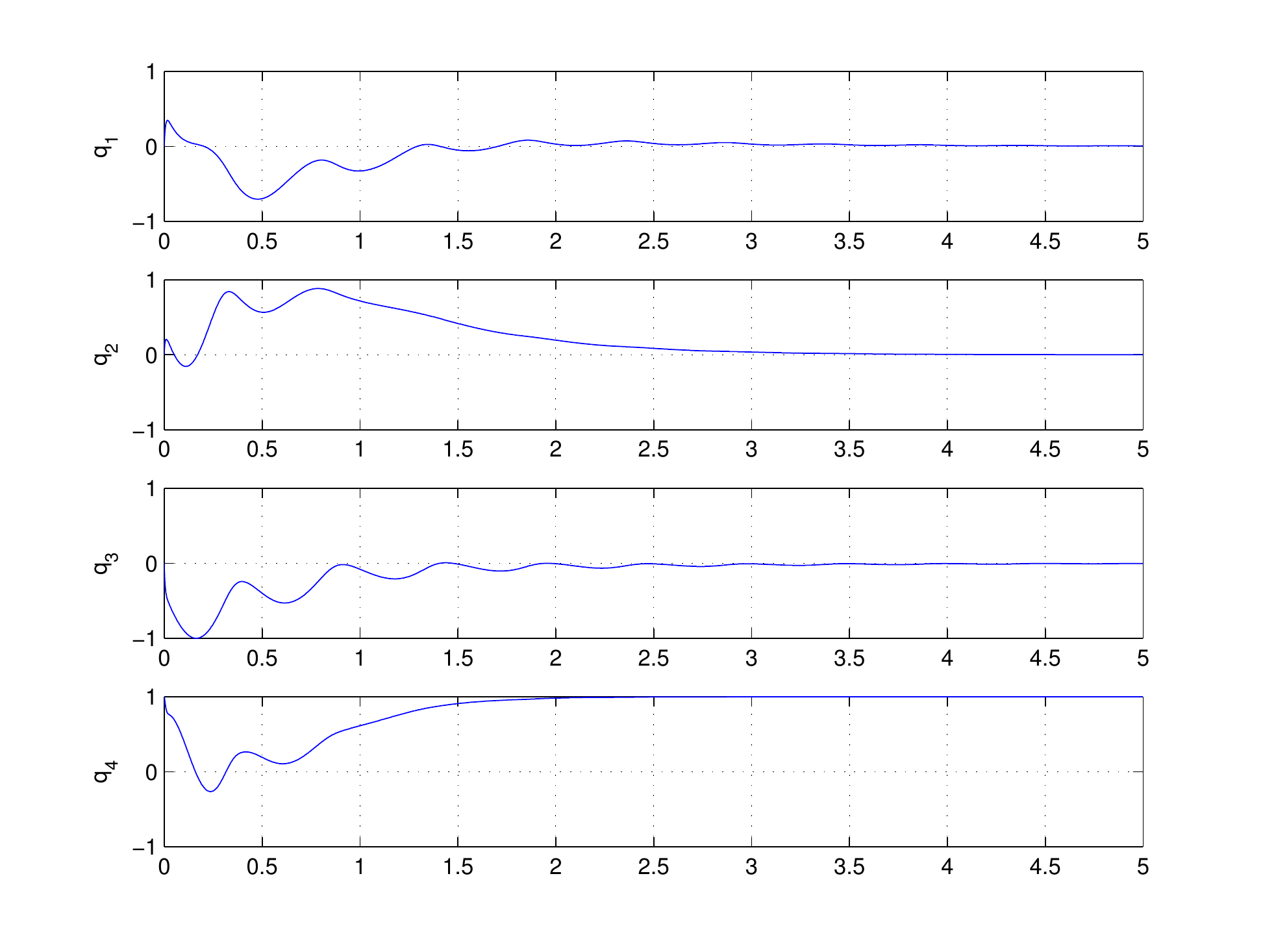}}\\[-6mm]
\subfigure{\includegraphics[width=8.4cm]{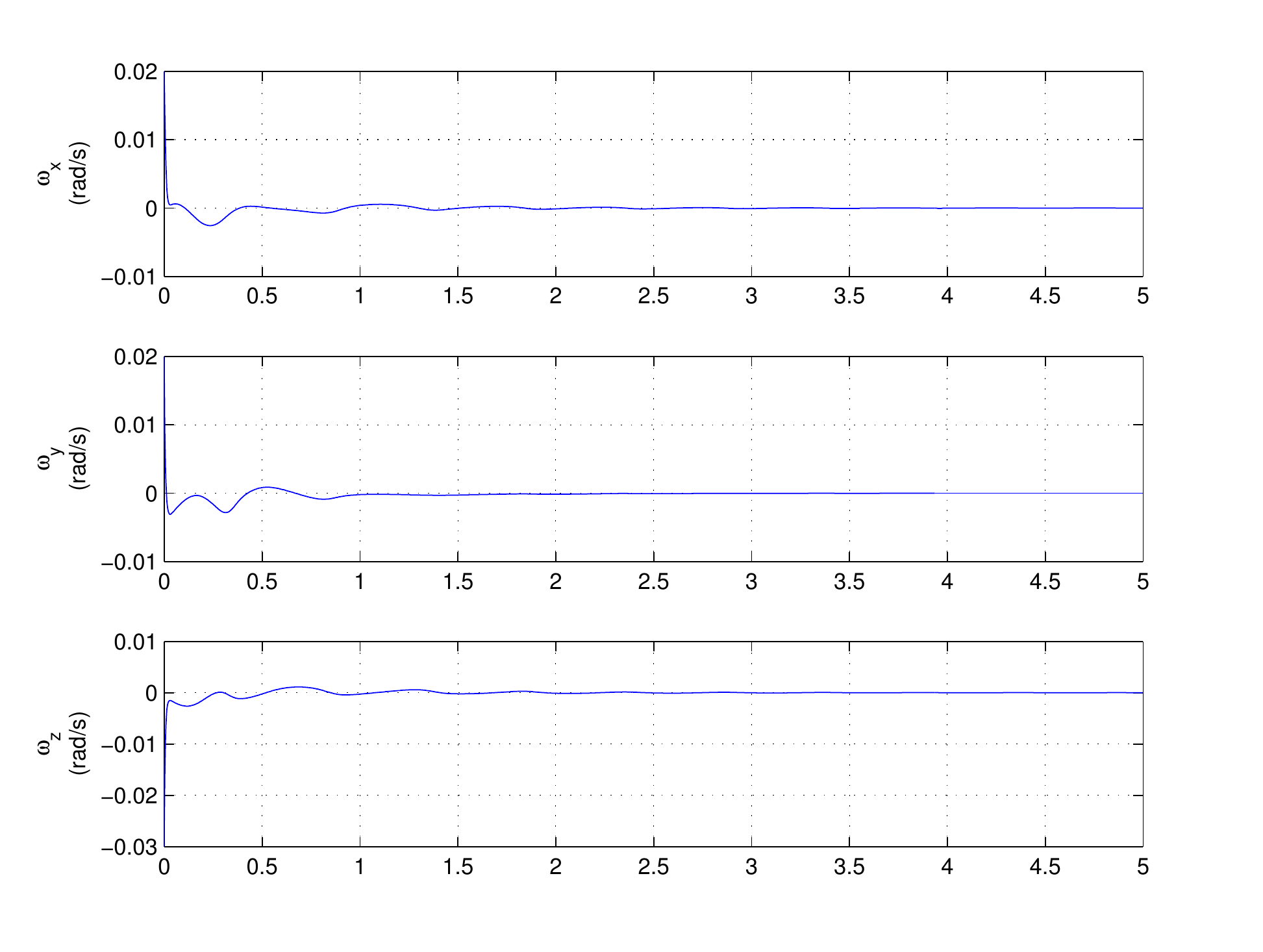}}\\[-5mm]
\subfigure{\includegraphics[width=8.4cm]{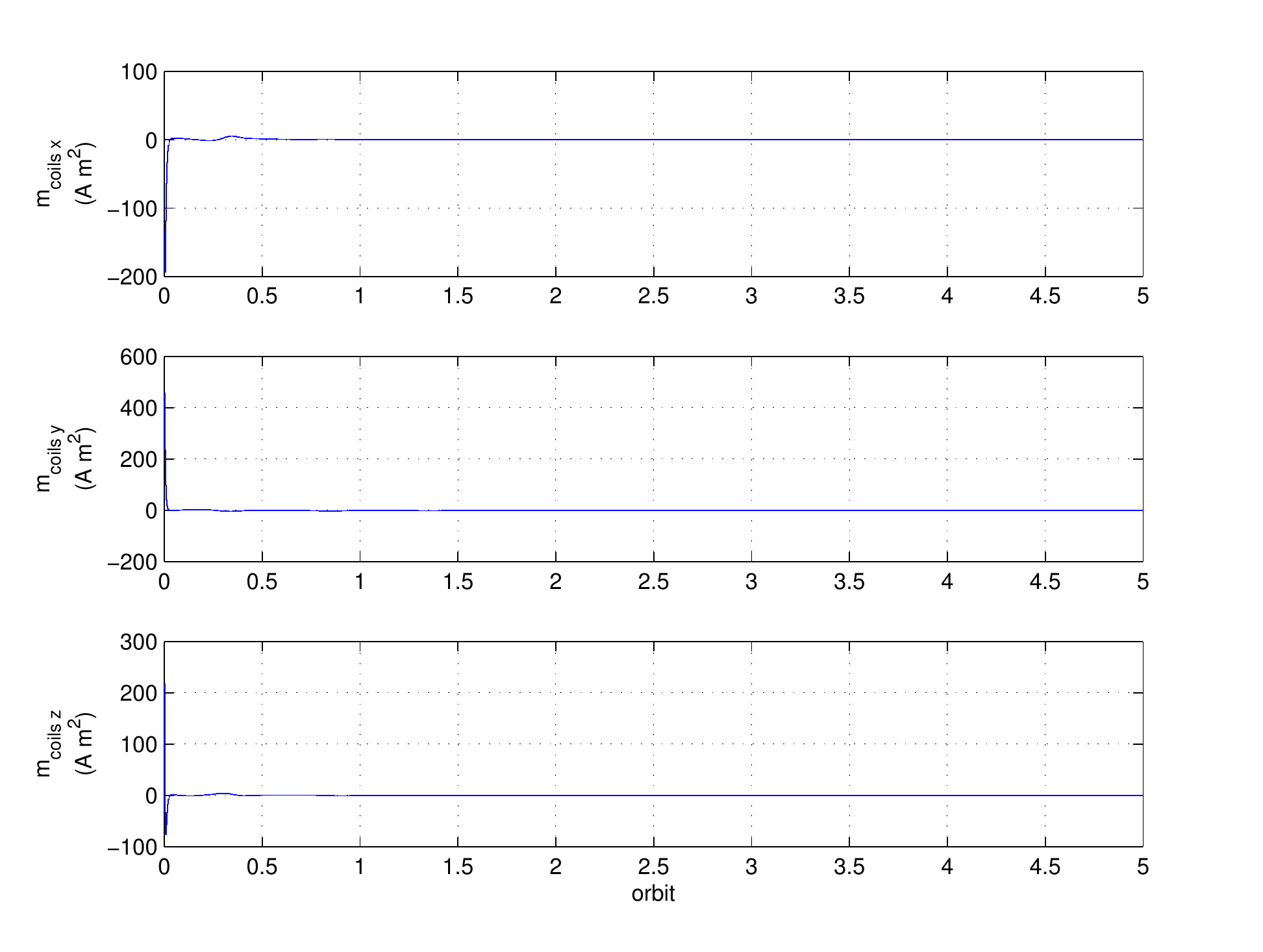}}\\[-6mm]
\caption{Evolutions of (\ref{nonlin_tv}) (\ref{m_coils_dc}) with $k_1=2 \ 10^{11}$, $k_2=3\ 10^{11}$, $T=20$ s, and $\epsilon=10^{-3}$.}
\label{evolutions_state_fdbck}
\end{figure}

A similar approach has been followed to select parameters and sampling interval for output-feedback (\ref{output_feedback_law_discrete}); the obtained results validate Proposition \ref{main_prop_output}, but they are omitted for lack of space.

\section{Conclusion}

In this work two magnetic control laws for spacecraft attitude stabilizations have been presented. Both magnetic control laws possess the feature of generating piecewise-constant magnetic dipole moments; the latter aspect makes them easier to implement than continuos-time control laws; in fact, implementation of  continuos-time laws requires a discretization process,  in which the discretization interval is usually selected by trial and error; the latter trial and error process in not necessary using the approach here proposed since the sampling interval can be determined in a systematic way.

\section*{Acknowledgment}

The author acknowledges Prof. J. L. Mancilla-Aguilar for his help and Prof. A. Nascetti for fruitful discussions.



\bibliographystyle{IEEEtran}
\bibliography{IEEEabrv,biblio}

%
%
%
%
%

\end{document}